\documentclass[12pt,abstract]{scrartcl}
\usepackage{graphicx}
\usepackage{amsmath,amsthm}
\makeatletter
\renewcommand{\theenumii}{\@roman\c@enumii}
\makeatother

\usepackage{xcolor}
\usepackage{natbib}
\newtheorem{lemma}{Lemma}
\newtheorem{theorem}{Theorem}
\newtheorem{assumption}{Assumption}

\usepackage{amssymb}
\usepackage{mathptmx}
\begin{document}

\title{Minimax theorem and Nash equilibrium of symmetric multi-players zero-sum game with two strategic variables}

\author{%
Masahiko Hattori\thanks{mhattori@mail.doshisha.ac.jp}\\[.01cm]
Faculty of Economics, Doshisha University,\\
Kamigyo-ku, Kyoto, 602-8580, Japan,\\[.1cm]
Atsuhiro Satoh\thanks{atsatoh@hgu.jp}\\[.01cm]
Faculty of Economics, Hokkai-Gakuen University,\\[.02cm]
Toyohira-ku, Sapporo, Hokkaido, 062-8605, Japan,\\[.01cm]
\textrm{and} \\[.1cm]
Yasuhito Tanaka\thanks{yasuhito@mail.doshisha.ac.jp}\\[.01cm]
Faculty of Economics, Doshisha University,\\
Kamigyo-ku, Kyoto, 602-8580, Japan.\\}

\date{}

\maketitle
\thispagestyle{empty}

\vspace{-1.4cm}

\begin{abstract}
We consider a symmetric multi-players zero-sum game with two strategic variables. There are $n$ players, $n\geq 3$. Each player is denoted by $i$. Two strategic variables are $t_i$ and $s_i$, $i\in \{1, \dots, n\}$. They are related by invertible functions. Using the minimax theorem by \cite{sion} we will show that Nash equilibria in the following states are equivalent.
\begin{enumerate}
	\item All players choose $t_i,\ i\in \{1, \dots, n\}$, (as their strategic variables).
	\item Some players choose $t_i$'s and the other players choose $s_i$'s.
	\item All players  choose $s_i,\ i\in \{1, \dots, n\}$.
\end{enumerate}

\end{abstract}

\begin{description}
	\item[Keywords:] symmetric multi-players zero-sum game, Nash equilibrium, two strategic variables
\end{description}

\begin{description}
	\item[JEL Classification:] C72
\end{description}

\section{Introduction}
We consider a symmetric multi-players zero-sum game with two strategic variables. There are $n$ players, $n\geq 3$. Each player is denoted by $i$. Two strategic variables are $t_i$ and $s_i$, $i\in \{1, \dots, n\}$. They are related by invertible functions. Using the minimax theorem by \cite{sion} we will show that Nash equilibria in the following states are equivalent.
\begin{enumerate}
	\item All players choose $t_i,\ i\in \{1, \dots, n\}$, (as their strategic variables).
	\item Some players choose $t_i$'s and the other players choose $s_i$'s.
	\item All players  choose $s_i,\ i\in \{1, \dots, n\}$.
\end{enumerate}

In the next section we present a model of this paper and prove some preliminary results which are variations of Sion's minimax theorem. In Section 3 we will show the main results. An example of a multi-players zero-sum game with two strategic variables is a relative profit maximization game in an oligopoly with differentiated goods. See Section 4.

\section{The model and the  minimax theorem}

We consider a symmetric multi-players zero-sum game with two strategic variables. There are $n$ players, $n\geq 3$. Two strategic variables are $t_i$ and $s_i$, $i\in \{1, \dots, n\}$. $t_i$ is chosen from $T_i$ and $s_i$ is chosen from $S_i$. $T_i$ and $S_i$ are convex and compact sets in linear topological spaces, respectively, for each $i\in \{1, \dots, n\}$. We denote $N=\{1, \dots, n\}$. The relations of the strategic variables are represented by
\begin{equation*}
s_i=f_i(t_1, \dots, t_n),\ i\in N,
\end{equation*}
and
\begin{equation*}
t_i=g_i(s_1, \dots, s_n),\ i\in N.
\end{equation*}
$f_i(t_1, \dots, t_n)$ and $g_i(s_1, \dots, s_n)$ are continuous invertible functions, and so they are one-to-one and onto functions. Let $M=\{1, \dots, m\},\ 0\leq m\leq n$, be a subset of $N$, and denote $N-M=\{m+1, \dots, n\}$. When $n-m$ players in $N-M$ choose $s_i$'s, $t_i$'s for them are determined according to
\[
\left\{
\begin{array}{l}
t_{m+1}=g_{m+1}(f_1(t_1, \dots, t_m, t_{m+1}, \dots, t_n),\dots, f_m(t_1, \dots, t_m, t_{m+1}, \dots, t_n), s_{m+1}, \dots, s_n)\\
\dots\\
t_n=g_n(f_1(t_1, \dots, t_m, t_{m+1}, \dots, t_n),\dots, f_m(t_1, \dots, t_m, t_{m+1}, \dots, t_n), s_{m+1}, \dots, s_n).
\end{array}\right.
\]
We denote these $t_i$'s by $t_i(t_1, \dots, t_m, s_{m+1}, \dots, s_n)$.

When all players choose $s_i$'s, $i\in N$, $t_i$'s for them are determined according to
\[
\left\{
\begin{array}{l}
t_{1}=g_{1}(s_1, \dots, s_n)\\
\dots\\
t_n=g_n(s_1, \dots, s_n).
\end{array}\right.
\]
Denote these $t_i$'s by $t_i(s_1, \dots, s_n)$.

The payoff function of Player $i$ is $u_i,\ i\in N$. It is written as
\[u_i(t_1, \dots, t_n).\]
We assume 
\begin{quote}
$u_i:T_1\times \dots\times T_n\Rightarrow \mathbb{R}$ for each $i\in N$ is continuous on $T_1\times \dots \times T_n$. Thus, it is continuous on $S_1\times \dots \times S_n$ through $f_i,\ i\in N$. It is quasi-concave on $T_i$ and $S_i$ for a strategy of each other player, and quasi-convex on $T_j,\ j\neq i$ and $S_j,\ j\neq i$ for each $t_i$ and $s_i$.
\end{quote}
We do not assume differentiability of the payoff functions.

Symmetry of the game means that the payoff functions of all players are symmetric and in the payoff function of each Player $i$, Players $j$ and $k,\ j, k\neq i$, are interchangeable. $f_i$'s and $g_i$'s are symmetric. Since the game is a zero-sum game, the sum of the values of the payoff functions of the players is zero.  All $T_i$'s are identical, and all $S_i$'s are identical. Denote them by $T$ and $S$.

Sion's minimax theorem (\cite{sion}, \cite{komiya}, \cite{kind}) for a continuous function is stated as follows.
\begin{lemma}
Let $X$ and $Y$ be non-void convex and compact subsets of two linear topological spaces, and let $f:X\times Y \rightarrow \mathbb{R}$ be a function that is continuous and quasi-concave in the first variable and continuous and quasi-convex in the second variable. Then
\[\max_{x\in X}\min_{y\in Y}f(x,y)=\min_{y\in Y}\max_{x\in X}f(x,y).\] \label{l1}
\end{lemma}
We follow the description of Sion's theorem in \cite{kind}.

Applying this lemma to the situation of this paper such that $m$ players choose $t_i$'s and $n-m$ players choose $s_i$'s as their strategic variables, we have the following relations.
\begin{align*}
\max_{t_i\in T}\min_{t_j\in T}u_i(t_i, t_j, \mathbf{t}_k, \mathbf{t}_l)=\min_{t_j\in T}\max_{t_i\in T}u_i(t_i, t_j, \mathbf{t}_k, \mathbf{t}_l).
\end{align*}
\begin{align*}
\max_{t_i\in T}\min_{s_j\in S}u_i(t_i, t_j(t_i, s_j,\mathbf{t}_k, \mathbf{t}_l), \mathbf{t}_k, \mathbf{t}_l)=\min_{s_j\in S}\max_{t_i\in T}u_i(t_i, t_j(t_i, s_j,\mathbf{t}_k, \mathbf{t}_l), \mathbf{t}_k, \mathbf{t}_l),
\end{align*}
where $\mathbf{t}_k$ is a vector of $t_k,\ k\in M$, of the players other than Players $i$ and $j$ who choose $t_k$'s as their strategic variables. On the other hand, $\mathbf{t}_l$ is a vector of $t_l,\ l\in N-M$, of the players other than Player  $j$ who choose $s_l$'s as their strategic variables. Also, relations which are symmetric to them hold. $u_i(t_i, t_j, \mathbf{t}_k, \mathbf{t}_l)$ is the payoff of Player $i$ when Players $i$ and $j$ choose $t_i$ and $t_j$. On the other hand, $u_i(t_i, t_j(t_i, s_j, \mathbf{t}_k, \mathbf{t}_l), \mathbf{t}_k, \mathbf{t}_l)$ means the payoff of Player $i$ when he chooses $t_i$ and Player $j$ chooses $s_j$.

Further we show the following results.
\begin{lemma}
\begin{align*}
&\max_{t_j\in T}\min_{t_i\in T}u_j(t_i, t_j, \mathbf{t}_k, \mathbf{t}_l)=\max_{s_j\in S}\min_{t_i\in T}u_j(t_i, t_j(t_i, s_j, \mathbf{t}_k, \mathbf{t}_l), \mathbf{t}_k, \mathbf{t}_l)\\
&=\min_{t_i\in T}\max_{s_j\in S}u_j(t_i, t_j(t_i, s_j, \mathbf{t}_k, \mathbf{t}_l), \mathbf{t}_k, \mathbf{t}_l)=\min_{t_i\in T}\max_{t_j\in T}u_j(t_i, t_j,\mathbf{t}_k, \mathbf{t}_l),
\end{align*}
\label{l2}
\label{lemma2}
\end{lemma}
$u_j(t_i, t_j, \mathbf{t}_k, \mathbf{t}_l)$ is the payoff of Player $j$ when Players $i$ and $j$ choose $t_i$ and $t_j$. On the other hand, $u_j(t_i, t_j(t_i, s_j, \mathbf{t}_k, \mathbf{t}_l), \mathbf{t}_k, \mathbf{t}_l)$ means the payoff of Player $j$ when he chooses $s_j$ and Player $i$ chooses $t_i$.
\begin{proof}
$\min_{t_i\in T}u_j(t_i, t_j(t_i, s_j, \mathbf{t}_k, \mathbf{t}_l), \mathbf{t}_k, \mathbf{t}_l)$ is the minimum of $u_j$ with respect to $t_i$ given $s_j$. Let $\tilde{t}_i(s_j)=\arg\min_{t_i\in T}u_j(t_i, t_j(t_i, s_j, \mathbf{t}_k, \mathbf{t}_l), \mathbf{t}_k, \mathbf{t}_l)$, and fix the value of $t_j$ at
\begin{equation}
t_j^0=g_j(f_i(\tilde{t}_i(s_j), t_j^0, \mathbf{t}_k, \mathbf{t}_l),s_j, \mathbf{f}_k, \mathbf{s}_l),\label{tc}
\end{equation}
where $\mathbf{f}_k$ denotes a vector of the values of $s_k$'s of players who choose $t_k$'s, and $\mathbf{s}_l$ denotes a vector of the values of $s_l$'s of players who choose $s_l$'s. Then, we have
\begin{align*}
\min_{t_i\in T}u_j(t_i, t_j^0,\mathbf{t}_k, \mathbf{t}_l)\leq u_j(\tilde{t}_i(s_j), t_j^0, \mathbf{t}_k, \mathbf{t}_l)=\min_{t_i\in T}u_j(t_i,t_j(t_i,s_j,\mathbf{t}_k,\mathbf{t}_l),\mathbf{t}_k,\mathbf{t}_l),
\end{align*}
where $\min_{t_i\in T}u_j(t_i, t_j^0,\mathbf{t}_k, \mathbf{t}_l)$ is the minimum of $u_j$ with respect to $t_i$ given the value of $t_j$ at $t_j^0$. We assume that $\tilde{t}_i(s_j)=\arg\min_{t_i\in T}u_j(u_j(t_i, t_j(t_i, s_j, \mathbf{t}_k, \mathbf{t}_l), \mathbf{t}_k, \mathbf{t}_l))$ is single-valued. By the maximum theorem and continuity of $u_j$, $\tilde{t}_i(s_j)$ is continuous. Then, any value of $t_j^0$ can be realized by appropriately choosing $s_j$ according to (\ref{tc}). Therefore,
\begin{equation}
\max_{t_j\in T}\min_{t_i\in T}u_j(t_i, t_j,\mathbf{t}_k, \mathbf{t}_l)\leq \max_{s_j\in S}\min_{t_i\in T}u_j(t_i,t_j(t_i,s_j,\mathbf{t}_k,\mathbf{t}_l),\mathbf{t}_k,\mathbf{t}_l).\label{4-1}
\end{equation}

On the other hand, $\min_{t_i\in T}u_j(t_i, t_j, \mathbf{t}_k, \mathbf{t}_l)$ is the minimum of $u_j$ with respect to $t_i$ given $t_j$. Let $\tilde{t}_i(t_j)=\arg\min_{t_i\in T}u_j(t_i, t_j,\mathbf{t}_k, \mathbf{t}_l)$, and fix the value of $s_j$ at
\begin{equation}
s_j^0=f_j(\tilde{t}_i(t_j),t_j,\mathbf{t}_k, \mathbf{t}_l).\label{sc}
\end{equation}
Then, we have
\begin{align*}
\min_{t_i\in T}u_j(t_i, t_j(t_i,s_j^0,\mathbf{t}_k, \mathbf{t}_l),\mathbf{t}_k, \mathbf{t}_l)\leq u_j(\tilde{t}_i(t_j),t_j(\tilde{t}_i(t_j),s_j^0,\mathbf{t}_k, \mathbf{t}_l),\mathbf{t}_k, \mathbf{t}_l)=\min_{t_i\in T}u_j(t_i, t_j, \mathbf{t}_k, \mathbf{t}_l),
\end{align*}
where $\min_{t_i\in T}u_j(t_i, t_j(t_i,s_j^0,\mathbf{t}_k, \mathbf{t}_l),\mathbf{t}_k, \mathbf{t}_l)$ is the minimum of $u_j$ with respect to $t_i$ given the value of $s_j$ at $s_j^0$. We assume that $\tilde{t}_i(t_j)=\arg\min_{t_i\in T}u_j(t_i, t_j,\mathbf{t}_k, \mathbf{t}_l)$ is single-valued. By the maximum theorem and continuity of $u_j$, $\tilde{t}_i(t_j)$ is continuous. Then, any value of $s_j^0$ can be realized by appropriately choosing $t_j$ according to (\ref{sc}). Therefore,
\begin{equation}
\max_{s_j\in S}\min_{t_i\in T}u_j(t_i,t_j(t_i, s_j, \mathbf{t}_k, \mathbf{t}_l) \mathbf{t}_k, \mathbf{t}_l)\leq \max_{t_j\in T}\min_{t_i\in T}u_j(t_i, t_j,\mathbf{t}_k, \mathbf{t}_l).\label{4-2}
\end{equation}
Combining (\ref{4-1}) and (\ref{4-2}), we get
\[\max_{s_j\in S}\min_{t_i\in T}u_j(t_i,t_j(t_i, s_j, \mathbf{t}_k, \mathbf{t}_l) \mathbf{t}_k, \mathbf{t}_l)=\max_{t_j\in T}\min_{t_i\in T}u_j(t_i, t_j,\mathbf{t}_k, \mathbf{t}_l).\]
Since any value of $s_j$ can be realized by appropriately choosing $t_j$, we have
\[\max_{s_j\in S}u_j(t_i,t_j(t_i,s_j,\mathbf{t}_k, \mathbf{t}_l),\mathbf{t}_k, \mathbf{t}_l)=\max_{t_j\in T}u_j(t_i, t_j, \mathbf{t}_k, \mathbf{t}_l).\]
Thus,
\[\min_{t_i\in T}\max_{s_j\in S}u_j(t_i,t_j(t_i,s_j,\mathbf{t}_k, \mathbf{t}_l),\mathbf{t}_k, \mathbf{t}_l)=\min_{t_i\in T}\max_{t_j\in T}u_j(t_i, t_j, \mathbf{t}_k, \mathbf{t}_l).\]
Therefore,
\begin{align*}
&\max_{t_j\in T}\min_{t_i\in T}u_j(t_i, t_j,\mathbf{t}_k, \mathbf{t}_l)=\max_{s_j\in S}\min_{t_i\in T}u_j(t_i, t_j(t_i,s_j,\mathbf{t}_k, \mathbf{t}_l),\mathbf{t}_k, \mathbf{t}_l)\\
=&\min_{t_i\in T}\max_{s_j\in S}u_j(t_i,t_j(t_i,s_j,\mathbf{t}_k, \mathbf{t}_l),\mathbf{t}_k, \mathbf{t}_l)=\min_{t_i\in T}\max_{t_j\in T}u_j(t_i,t_j,\mathbf{t}_k, \mathbf{t}_l).
\end{align*}
\end{proof}
\begin{lemma}
\begin{align*}
&\min_{t_j\in T}\max_{t_i\in T}u_i(t_i, t_j, \mathbf{t}_k, \mathbf{t}_l)=\min_{s_j\in S}\max_{t_i\in T}u_i(t_i, t_j(t_i, s_j, \mathbf{t}_k, \mathbf{t}_l), \mathbf{t}_k, \mathbf{t}_l)\\
&=\max_{t_i\in T}\min_{s_j\in S}u_i(t_i, t_j(t_i, s_j, \mathbf{t}_k, \mathbf{t}_l), \mathbf{t}_k, \mathbf{t}_l)=\max_{t_i\in T}\min_{t_j\in T}u_i(t_i, t_j,\mathbf{t}_k, \mathbf{t}_l),
\end{align*}
\label{l3}
\end{lemma}
\begin{proof}
$\max_{t_i\in T}u_i(t_i,t_j(t_i, s_j, \mathbf{t}_k, \mathbf{t}_l),\mathbf{t}_k, \mathbf{t}_l)$ is the maximum of $u_i$ with respect to $t_i$ given $s_j$. Let $\bar{t}_i(s_j)=\arg\max_{t_i\in T}u_i(t_i,t_j(t_i, s_j, \mathbf{t}_k, \mathbf{t}_l),\mathbf{t}_k, \mathbf{t}_l)$, and fix the value of $t_j$ at
\begin{equation}
t_j^0=g_j(f_i(\bar{t}_i(s_j), t_j^0, \mathbf{t}_k, \mathbf{t}_l),s_j, \mathbf{f}_k, \mathbf{s}_l),\label{tc1}
\end{equation}
where $\mathbf{f}_k$ denotes a vector of the values of $s_k$'s of players who choose $t_k$'s, and $\mathbf{s}_l$ denotes a vector of the values of $s_l$'s of players who choose $s_l$'s. Then, we have
\begin{align*}
\max_{t_i\in T}u_i(t_i, t_j^0, \mathbf{t}_k, \mathbf{t}_l)\geq u_i(\bar{t}_i(s_j),t_j^0,\mathbf{t}_k, \mathbf{t}_l)=\max_{t_i\in T}u_i(t_i,t_j(t_i, s_j, \mathbf{t}_k, \mathbf{t}_l),\mathbf{t}_k, \mathbf{t}_l),
\end{align*}
where $\max_{t_i\in T}u_i(t_i, t_j^0,\mathbf{t}_k, \mathbf{t}_l)$ is the maximum of $u_i$ with respect to $t_i$ given the value of $t_j$ at $t_j^0$. We assume that $\bar{t}_i(s_j)=\arg\max_{t_i\in T}u_i(t_i,t_j(t_i,s_j,\mathbf{t}_k, \mathbf{t}_l),\mathbf{t}_k, \mathbf{t}_l)$ is single-valued. By the maximum theorem and continuity of $u_i$, $\bar{t}_i(s_j)$ is continuous. Then, any value of $t_j^0$ can be realized by appropriately choosing $s_j$ according to (\ref{tc1}). Therefore,
\begin{equation}
\min_{t_j\in T}\max_{t_i\in T}u_i(t_i, t_j,\mathbf{t}_k, \mathbf{t}_l)\geq \min_{s_j\in S}\max_{t_i\in T}u_i(t_i, t_j(t_i, s_j, \mathbf{t}_k, \mathbf{t}_l),\mathbf{t}_k, \mathbf{t}_l).\label{4-11}
\end{equation}

On the other hand, $\max_{t_i\in T}u_i(t_i, t_j,\mathbf{t}_k, \mathbf{t}_l)$ is the maximum of $u_i$ with respect to $t_i$ given $t_j$. Let $\bar{t}_i(t_j)=\arg\max_{t_i\in T}u_i(t_i, t_j,\mathbf{t}_k, \mathbf{t}_l)$, and fix the value of $s_j$ at
\begin{equation}
s_j^0=f_j(\bar{t}_i(t_j), t_j,\mathbf{t}_k, \mathbf{t}_l).\label{sc1}
\end{equation}
Then, we have
\begin{align*}
\max_{t_i\in T}u_i(t_i,t_j(t_i,s_j^0,\mathbf{t}_k, \mathbf{t}_l),\mathbf{t}_k, \mathbf{t}_l)\geq u_i(\bar{t}_i(s_j),t_j(t_i,s_j^0,\mathbf{t}_k, \mathbf{t}_l),\mathbf{t}_k, \mathbf{t}_l)=\max_{t_i\in T}u_i(t_i, t_j,\mathbf{t}_k, \mathbf{t}_l),
\end{align*}
where $\max_{t_i\in T}u_i(t_i,t_j(t_i,s_j,\mathbf{t}_k, \mathbf{t}_l),\mathbf{t}_k, \mathbf{t}_l)$ is the maximum of $u_i$ with respect to $t_i$ given the value of $s_j$ at $s_j^0$. We assume that $\bar{t}_i(t_j)=\arg\max_{t_i\in T}u_i(t_i, t_j,\mathbf{t}_k, \mathbf{t}_l)$ is single-valued. By the maximum theorem and continuity of $u_i$, $\bar{t}_i(t_j)$ is continuous. Then, any value of $s_j^0$ can be realized by appropriately choosing $t_j$ according to (\ref{sc1}). Therefore,
\begin{equation}
\min_{s_j\in S}\max_{t_i\in T}u_i(t_i, t_j(t_i,s_j,\mathbf{t}_k, \mathbf{t}_l),\mathbf{t}_k, \mathbf{t}_l)\geq \min_{t_j\in T}\max_{t_i\in T}u_i(t_i, t_j,\mathbf{t}_k, \mathbf{t}_l).\label{4-21}
\end{equation}
Combining (\ref{4-11}) and (\ref{4-21}), we get
\[\min_{s_j\in S}\max_{t_i\in T}u_i(t_i, t_j(t_i,s_j,\mathbf{t}_k, \mathbf{t}_l),\mathbf{t}_k, \mathbf{t}_l)=\min_{t_j\in T}\max_{t_i\in T}u_i(t_i, t_j,\mathbf{t}_k, \mathbf{t}_l).\]
Since any value of $s_j$ can be realized by appropriately choosing $t_j$, we have
\begin{equation*}
\min_{s_j\in S}u_i(t_i,t_j(t_i,s_j,\mathbf{t}_k, \mathbf{t}_l),\mathbf{t}_k, \mathbf{t}_l)=\min_{t_j\in T}u_i(t_i, t_j,\mathbf{t}_k, \mathbf{t}_l).
\end{equation*}
Thus,
\[\max_{t_i\in T}\min_{s_j\in S}u_i(t_i,t_j(t_i,s_j,\mathbf{t}_k, \mathbf{t}_l),\mathbf{t}_k, \mathbf{t}_l)=\max_{t_i\in T}\min_{t_j\in T}u_i(t_i, t_j,\mathbf{t}_k, \mathbf{t}_l).\]
Therefore,
\begin{align*}
&\min_{t_j\in T}\max_{t_i\in T}u_i(t_i, t_j,\mathbf{t}_k, \mathbf{t}_l)=\min_{s_j\in S}\max_{t_i\in T}u_i(t_i, t_j(t_i, s_j, \mathbf{t}_k, \mathbf{t}_l),\mathbf{t}_k, \mathbf{t}_l),\\
=&\max_{t_i\in T}\min_{s_j\in S}u_i(t_i, t_j(t_i,s_j,\mathbf{t}_k, \mathbf{t}_l),\mathbf{t}_k, \mathbf{t}_l)=\max_{t_i\in T}\min_{t_j\in T}u_i(t_i, t_j,\mathbf{t}_k, \mathbf{t}_l).
\end{align*}
\end{proof}

\section{The main results}

In this section we present the main results of this paper. First we show
\begin{theorem}
The equilibrium where all players choose $t_i$'s is equivalent to the equilibrium where one player (Player $j$) chooses $s_j$ and all other players choose $t_i$'s as their strategic variables.\label{t1}
\end{theorem}
\begin{proof}
\begin{enumerate}
	\item Consider a situation $(t_1,\dots, t_n)=(t,\dots,t)$, that is, all players choose the same value of $t_i$. Let
\[s^0(t)=f_i(t, \dots, t),\ i\in N.\]
By symmetry of the game
\[\max_{t_1\in T}u_1(t_1, t, \dots, t)=\dots =\max_{t_n\in T}u_n(t,\dots,t_n),\]
and
\[\arg\max_{t_1\in T}u_1(t_1, t, \dots, t)=\dots =\arg\max_{t_n\in T}u_n(t,\dots,t_n).\]
Consider the following function.
\[t\rightarrow \arg\max_{t_i\in T}u_i(t_i, t, \dots, t),\ i\in N.\]
Since this function is continuous and $T$ is compact, there exists a fixed point. Denote it by $t^*$. Then,
\[t^*\rightarrow \arg\max_{t_i\in T}u_i(t_i, t^*, \dots, t^*).\]
We have
\[\max_{t_i\in T}u_i(t_i, t^*, \dots, t^*)=0,\ \mathrm{for\ all}\ i\in N.\]

\item Because the game is zero-sum,
\[u_i(t_i, t^*, \dots, t^*)+\sum_{j=1, j\neq i}^nu_j(t_i,t^*, \dots,t^*)=0.\]
By symmetry
\[u_i(t_i, t^*, \dots, t^*)+(n-1)u_j(t_i,t^*, \dots,t^*)=0.\]
This means
\[u_i(t_i, t^*, \dots, t^*)=-(n-1)u_j(t_i,t^*, \dots,t^*).\]
and
\[\max_{t_i\in T}u_i(t_i, t^*, \dots, t^*)=-(n-1)\min_{t_i\in T}u_j(t_i, t^*, \dots, t^*).\]
From this we get
\[\arg\max_{t_i\in T}u_i(t_i, t^*, \dots, t^*)=\arg\min_{t_i\in T}u_j(t_i, t^*, \dots, t^*)=t^*.\]
We have
\[\max_{t_i\in T}u_i(t_i, t^*, \dots, t^*)=\min_{t_i\in T}u_j(t_i,t^*,\dots, t^*)=u_i(t^*, \dots, t^*)=0.\]
By symmetry
\[\max_{t_i\in T}u_i(t_i, t^*, \dots, t^*)=\min_{t_j\in T}u_i(t_j,t^*,\dots, t^*)=0.\]
Then,
\begin{align*}
&\min_{t_j\in T}\max_{t_i\in T}u_i(t_i, t_j, t^*,\dots, t^*)\leq \max_{t_i\in T}u_i(t_i, t^*, \dots, t^*)\\
&=\min_{t_j\in T}u_i(t_j, t^*,\dots,t^*)\leq \max_{t_i\in T}\min_{t_j\in T}u_i(t_i,t_j,t^*,\dots, t^*).
\end{align*}
From Lemma \ref{l3} we obtain
\begin{align}
&\min_{t_j\in T}\max_{t_i\in T}u_i(t_i, t_j, t^*,\dots, t^*)=\max_{t_i\in T}u_i(t_i, t^*,\dots, t^*)=\min_{t_j\in T}u_i(t_j, t^*,\dots, t^*)\label{l3-1}\\
&=\max_{t_i\in T}\min_{t_j\in T}u_i(t_i, t_j, t^*,\dots, t^*)=\min_{s_j\in S}\max_{t_i\in T}u_i(t_i, t_j(t_i, s_j,t^*,\dots, t^*), t^*,\dots, t^*)\notag \\
&=\max_{t_i\in T}\min_{s_j\in S}u_i(t_i, t_j(t_i, s_j, t^*,\dots, t^*), t^*,\dots, t^*)=0.\notag
\end{align}

\item Since any value of $s_j$ can be realized by appropriately choosing $t_j$,
\begin{align}
&\min_{s_j\in S}u_i(t^*,t_j(t^*, s_j, t^*,\dots, t^*), \dots, t^*)=\min_{t_j\in T}u_i(t^*,t_j,t^* \dots, t^*)\label{z1}\\
&=u_i(t^*,\dots,t^*)=0.\notag
\end{align}
Then,
\[\arg\min_{s_j\in S}u_i(t^*,t_j(t^*, s_j, t^*,\dots, t^*), t^*, \dots, t^*)=s^0(t^*).\]
(\ref{l3-1}) and (\ref{z1}) mean
\begin{align}
&\min_{s_j\in S}\max_{t_i\in T}u_i(t_i, t_j(t^*, s_j, t^*,\dots, t^*), t^*, \dots, t^*)\label{z2}\\
&=\min_{s_j\in S}u_i(t^*,t_j(t^*, s_j, t^*,\dots, t^*), t^*, \dots, t^*)=0.\notag
\end{align}
And we have
\[\max_{t_i\in T}u_i(t_i, t_j(t^*, s_j, t^*,\dots, t^*), t^*, \dots, t^*)\geq u_i(t^*, t_j(t^*, s_j, t^*,\dots, t^*), t^*, \dots, t^*).\]


Then,
\begin{align*}
&\arg\min_{s_j\in S}\max_{t_i\in T}u_i(t_i, t_j(t^*, s_j, t^*,\dots, t^*), t^*, \dots, t^*)\\
&=\arg\min_{s_j\in S}u_i(t^*, t_j(t^*, s_j, t^*,\dots, t^*), t^*, \dots, t^*)=s^0(t^*).
\end{align*}
Note $s^0(t^*)=f(t^*,t^*,\dots,t^*)$. 

Thus, by (\ref{z2})
\begin{align*}
&\min_{s_j\in S}\max_{t_i\in T}u_i(t_i, t_j(t^*, s_j, t^*,\dots, t^*), t^*, \dots, t^*)=\max_{t_i\in T}u_i(t_i, t_j(t_i, s^0(t^*), t^*,\dots, t^*), t^*, \dots, t^*)\\
&=\min_{s_j\in S}u_i(t^*,t_j(t^*, s_j, t^*,\dots, t^*), t^*, \dots, t^*)=u_i(t^*, t_j(t^*, s^0(t^*), t^*,\dots, t^*), t^*, \dots, t^*)=0.
\end{align*}


Therefore,
\begin{equation}
\arg\max_{t_i\in T}u_i(t_i, t_j(t_i, s^0(t^*), t^*,\dots, t^*), t^*, \dots, t^*)=t^*.\label{t1-2}
\end{equation}
This holds for all $i\in N,\ i\neq j$.

On the other hand, because any value of $s_j$ is realized by appropriately choosing $t_j$,
\[\max_{s_j\in S}u_j(t^*, t_j(t^*, s_j, t^*, \dots, t^*))=\max_{t_j\in T}u_j(t^*, t_j, t^*, \dots, t^*)=u_j(t^*, \dots, t^*)=0.\]
Therefore,
\begin{equation}
\arg\max_{s_j\in S}u_j(t^*, t_j(t^*, s_j, t^*, \dots, t^*))=s^0(t^*).\label{t1-1}
\end{equation}

From (\ref{t1-2}) and (\ref{t1-1}), $(t^*,s^0(t^*),t^*, \dots, t^*)$ is a Nash equilibrium which is equivalent to $(t^*,\dots,t^*)$. $(t^*,s^0(t^*),t^*, \dots, t^*)$ denotes an equilibrium where $t_i=t^*,\ s_j=s^0(t^*)$ and $t_k=t^*$ for $k\neq i, j$. 
\end{enumerate}
\end{proof}

Consider a Nash equilibrium where $m$ players choose $t^*$ and $n-m$ players choose $s^0(t^*)$. Let $\mathbf{t}_k$ be a vector of $t_k,\ k\in M$, of players other than $i$ and $j$ who choose $t_k$'s as their strategic variables; $\mathbf{t}_l$ be a vector of $t_l,\ l\in N-M$, of players who choose $s_l$'s as their strategic variables. These expressions mean that $t_i=t_j=t^*$; each $t_k=t^*$ and each $s_l=s^0(t^*)$. We write such an equilibrium as $(t_i, t_j, \mathbf{t}_k, \mathbf{t}_l,)=(t^*, t^*, \mathbf{t}^*_k, \mathbf{t}^*_l,)$. In the next theorem, based on Assumption \ref{as1}, we will show that such a Nash equilibrium is equivalent to a Nash equilibrium where $m-1$ players choose $t^*$ and $n-m+1$ players choose $s^0(t^*)$.

Now we assume
\begin{assumption}
At the equilibrium where $m$ players choose $t^*$ and $n-m$ players choose $s^0(t^*)$, the responses of $u_k$ and $u_l$ to a small change in $t_i$ have the same sign. 
\label{as1}
\end{assumption}
$u_k$ is the payoff of each player, other than $i$, whose strategic variable is $t_k$, and $u_l$ is the payoff of each player whose strategic variable is $s_l$.

\begin{quote}
When $t_i=t^*$ and $s_l=s^0(t^*)$ for $i\in M,\ l\in N-M$, we have $t_l=t^*$ for all $l\in N-M$. $u_k,\ k\in M\setminus i$ and $u_l,\ l\in N-M$ respond to a change in $t_i,\ i\in M$ given $t_k,\ k\in M\setminus i$ and $s_l,\ l\in N-M$. Since $s_k,\ k\in M\setminus i$ and $t_l,\ l\in N-M$ are not constant, the responses of $u_k,\ k\in M\setminus i$ and the responses of $u_l,\ l\in N-M$ to a change in $t_i,\ i\in M$ may be different. However, because all $t_i$'s are equal and all $u_i$'s for $i\in N$ are equal at the equilibrium, we may assume that the responses of $u_k,\ i\in M\setminus i$ and the responses of $u_l,\ l\in N-M$ to a change in $t_i,\ i\in M$ have the same sign in a sufficiently small neighborhood of the equilibrium.
\end{quote}

Using this assumption we show the following result.

\begin{theorem}
The equilibrium where $m$, $2\leq m\leq n-1$,  players choose $t_i$'s and $n-m$ players choose $s_i$'s as their strategic variables is equivalent to the equilibrium where $m-1$ players choose $t_i$'s and $n-m+1$ players choose $s_i$'s as their strategic variables.\label{t2}
\end{theorem}
\begin{proof}
Suppose that Player $i$ chooses $t_i$ in both equilibria, but Player $j$ chooses $t_j$ when $m$ players choose $t_i$'s and he chooses $s_j$ when $m-1$ players choose $t_i$'s. Then,
\[\arg\max_{t_i\in T}u_i(t_i, t^*, \mathbf{t}^*_k,\mathbf{t}^*_l)=\arg\max_{t_j\in T}u_j(t^*, t_j, \mathbf{t}^*_k,\mathbf{t}^*_l)=t^*.\]
Since any value of $t_j$ is realized by appropriately choosing $s_j$, we get
\[\max_{s_j\in S}u_j(t^*, t_j(t^*, s_j, \mathbf{t}^*_k,\mathbf{t}^*_l), \mathbf{t}^*_k,\mathbf{t}^*_l)=\max_{t_j\in T}u_j(t^*, t_j, \mathbf{t}^*_k,\mathbf{t}^*_l)=u_j(t^*, t^*, \mathbf{t}^*_k,\mathbf{t}^*_l),\]
and
\begin{equation}
\arg\max_{s_j\in S}u_j(t^*, t_j(t^*, s_j, \mathbf{t}^*_k,\mathbf{t}^*_l), \mathbf{t}^*_k,\mathbf{t}^*_l)=s^0(t^*).\label{n2-1}
\end{equation}
Since the game is zero-sum,
\begin{align*}
u_i(t_i, t^*,\mathbf{t}^*_k,\mathbf{t}^*_l)+u_j(t_i, t^*,\mathbf{t}^*_k,\mathbf{t}^*_l)+(m-2)u_k(t_i, t^*,\mathbf{t}^*_k,\mathbf{t}^*_l)+(n-m)u_l(t_i,t^*,\mathbf{t}^*_k,\mathbf{t}^*_l)=0,\notag
\end{align*}
and so
\begin{align*}
u_i(t_i, t^*,\mathbf{t}^*_k,\mathbf{t}^*_l)=-[u_j(t_i, t^*,\mathbf{t}^*_k,\mathbf{t}^*_l)+(m-2)u_k(t_i, t^*,\mathbf{t}^*_k,\mathbf{t}^*_l)+(n-m)u_l(t_i,t^*,\mathbf{t}^*_k,\mathbf{t}^*_l)],\notag
\end{align*}
$u_k$ denotes the payoff of each player who chooses $t_k$ as its strategic variable. Player $j$ is one of such players. $u_l$ denotes the payoff of each player who chooses $s_l$ as its strategic variable. Then, we obtain
\begin{align*}
u_i(t_i, t^*,\mathbf{t}^*_k,\mathbf{t}^*_l)=-[(m-1)u_j(t_i, t^*,\mathbf{t}^*_k,\mathbf{t}^*_l)+(n-m)u_l(t_i,t^*,\mathbf{t}^*_k,\mathbf{t}^*_l)].\notag
\end{align*}
Thus,
\begin{align*}
\max_{t_i\in T}u_i(t_i, t^*,\mathbf{t}^*_k,\mathbf{t}^*_l)=-\min_{t_i\in T}[(m-1)u_j(t_i, t^*,\mathbf{t}^*_k,\mathbf{t}^*_l)+(n-m)u_l(t_i,t^*,\mathbf{t}^*_k,\mathbf{t}^*_l)].\notag
\end{align*}
By Assumption \ref{as1} since $u_i(t_i, t^*,\mathbf{t}^*_k,\mathbf{t}^*_l)\leq 0$,
\[u_j(t_i, t^*,\mathbf{t}^*_k,\mathbf{t}^*_l)\geq 0,\ u_l(t_i,t^*,\mathbf{t}^*_k,\mathbf{t}^*_l)\geq 0,\]
in any neighborhood of $(t^*, t^*,\mathbf{t}^*_k,\mathbf{t}^*_l)$. Thus, we have
\begin{subequations}
\begin{equation*}
\min_{t_i\in T}u_j(t_i, t^*,\mathbf{t}^*_k,\mathbf{t}^*_l)=0,\label{2-1}
\end{equation*}
\begin{equation}
\arg\min_{t_i\in T}u_j(t_i, t^*,\mathbf{t}^*_k,\mathbf{t}^*_l)=t^*,\label{2-2}
\end{equation}
\begin{equation*}
\min_{t_i\in T}u_l(t_i,t^*,\mathbf{t}^*_k,\mathbf{t}^*_l)=0,\label{2-3}
\end{equation*}
and
\begin{equation}
\arg\min_{t_i\in T}u_l(t_i,t^*,\mathbf{t}^*_k,\mathbf{t}^*_l)=t^*.\label{2-4}
\end{equation}
\end{subequations}
By symmetry
\begin{equation*}
\min_{t_j\in T}u_i(t^*,t_j,\mathbf{t}^*_k,\mathbf{t}^*_l)=0,
\end{equation*}
\begin{equation*}
\arg\min_{t_j\in T}u_i(t^*,t_j,\mathbf{t}^*_k,\mathbf{t}^*_l)=t^*.
\end{equation*}
Thus,
\[\max_{t_i\in T}u_i(t_i,t^*,\mathbf{t}^*_k,\mathbf{t}^*_l)=\min_{t_j\in T}u_i(t^*,t_j,\mathbf{t}^*_k,\mathbf{t}^*_l)=u_i(t^*,t^*,\mathbf{t}^*_k,\mathbf{t}^*_l)=0.\]
Then,
\begin{align*}
\min_{t_j\in T}\max_{t_i\in T}u_i(t_i,t_j,\mathbf{t}^*_k,\mathbf{t}^*_l)\leq \max_{t_i\in T}u_i(t_i,t^*,\mathbf{t}^*_k,\mathbf{t}^*_l)=\min_{t_j\in T}u_i(t^*,t_j,\mathbf{t}^*_k,\mathbf{t}^*_l)\leq \max_{t_i\in T}\min_{t_j\in T}u_i(t_i,t_j,\mathbf{t}^*_k,\mathbf{t}^*_l).
\end{align*}
From Lemma \ref{l3}
\begin{align}
&\min_{t_j\in T}\max_{t_i\in T}u_i(t_i,t_j,\mathbf{t}^*_k,\mathbf{t}^*_l)=\max_{t_i\in T}u_i(t_i,t^*,\mathbf{t}^*_k,\mathbf{t}^*_l)=\min_{t_j\in T}u_i(t^*,t_j,\mathbf{t}^*_k,\mathbf{t}^*_l)\label{4-22}\\
&=\max_{t_i\in T}\min_{t_j\in T}u_i(t_i,t_j,\mathbf{t}^*_k,\mathbf{t}^*_l)=\min_{s_j\in S}\max_{t_i\in T}u_i(t_i,t_j(t_i, s_j,\mathbf{t}^*_k,\mathbf{t}^*_l),\mathbf{t}^*_k,\mathbf{t}^*_l)\notag\\
&=\max_{t_i\in T}\min_{s_j\in S}u_i(t_i,t_j(t_i, s_j, \mathbf{t}^*_k,\mathbf{t}^*_l),\mathbf{t}^*_k,\mathbf{t}^*_l).\notag
\end{align}
Since any value of $t_j$ is realized by appropriately choosing $s_j$ given $s_i=s^0(t^*)$ for all $i\neq n$,
\begin{equation}
\min_{t_j\in T}u_i(t^*,t_j, \mathbf{t}^*_k,\mathbf{t}^*_l)=\min_{s_j\in S}u_i(t^*, t_j(t^*, s_j, \mathbf{t}^*_k,\mathbf{t}^*_l),\mathbf{t}^*_k,\mathbf{t}^*_l)=0.\label{l4-2}
\end{equation}
Thus,
\[\arg\min_{s_j\in S}u_i(t^*, t_j(t^*, s_j, \mathbf{t}^*_k,\mathbf{t}^*_l),\mathbf{t}^*_k,\mathbf{t}^*_l)=s^0(t^*).\]
From (\ref{4-22}) and (\ref{l4-2})
\begin{equation}
\min_{s_j\in T}\max_{t_i\in T}u_i(t^*,t_j(t^*, s_j,\mathbf{t}^*_k,\mathbf{t}^*_l),\mathbf{t}^*_k,\mathbf{t}^*_l)=\min_{s_j\in S}u_i(t^*,t_j(t^*,s_j,\mathbf{t}^*_k,\mathbf{t}^*_l),\mathbf{t}^*_k,\mathbf{t}^*_l)=0.\label{z3}
\end{equation}
And we have
\begin{equation*}
\max_{t_i\in T}u_i(t_i,t_j(t_i, s_j,\mathbf{t}^*_k,\mathbf{t}^*_l),\mathbf{t}^*_k,\mathbf{t}^*_l)\geq u_i(t_i,t_j(t_i, s_j,\mathbf{t}^*_k,\mathbf{t}^*_l),\mathbf{t}^*_k,\mathbf{t}^*_l).
\end{equation*}


Then,
\begin{align*}
\arg\min_{s_j\in S}\max_{t_i\in T}u_i(t_i,t_j(t_i, s_j,\mathbf{t}^*_k,\mathbf{t}^*_l),\mathbf{t}^*_k,\mathbf{t}^*_l)=\arg\min_{s_j\in S}u_i(t^*,t_j(t^*,s_j,\mathbf{t}^*_k,\mathbf{t}^*_l)),\mathbf{t}^*_k,\mathbf{t}^*_l)=s^0(t^*).
\end{align*}


By (\ref{z3}) we get
\begin{align*}
&\min_{s_j\in T}\max_{t_i\in T}u_i(t_i,t_j(t_i, s_j,\mathbf{t}^*_k,\mathbf{t}^*_l),\mathbf{t}^*_k,\mathbf{t}^*_l)=\max_{t_i\in T}u_i(t_i,t_j(t_i, s^0(t^*),\mathbf{t}^*_k,\mathbf{t}^*_l),\mathbf{t}^*_k,\mathbf{t}^*_l)\\
=&\min_{s_j\in S}u_i(t^*,t_j(t^*,s_j,\mathbf{t}^*_k,\mathbf{t}^*_l),\mathbf{t}^*_k,\mathbf{t}^*_l)=u_i(t^*,t_j(t^*,s^0(t^*),\mathbf{t}^*_k,\mathbf{t}^*_l)),\mathbf{t}^*_k,\mathbf{t}^*_l)=0.
\end{align*}


Therefore, 
\begin{equation}
\arg\max_{t_i\in T}u_i(t_i,t_j(t_i, s^0(t^*),\mathbf{t}^*_k,\mathbf{t}^*_l),\mathbf{t}^*_k,\mathbf{t}^*_l)=t^*.\label{n2}
\end{equation}
This holds for all $i\in N,\ i\neq j$.

From (\ref{n2-1}) and (\ref{n2}) $(t^*,t_j(t^*,s^0(t^*),\mathbf{t}^*_k,\mathbf{t}^*_l),\mathbf{t}^*_k,\mathbf{t}^*_l)$ is a Nash equilibrium which is equivalent to $(t^*,t^*,\mathbf{t}^*_k,\mathbf{t}^*_l)$, and hence it is equivalent to $(t^*,\dots,t^*)$. Note that $i$ and $j$ are arbitrary. 
\end{proof}

By mathematical induction this theorem means that the Nash equilibrium where one player chooses $t_i$ and $n-1$ players choose $s_i$'s as their strategic variables is equivalent to the Nash equilibrium where all players choose $t_i$'s as their strategic variables. Suppose that in the former equilibrium only Player $n$ chooses $t_n$ and the other players choose $s_i$'s as their strategic variables. Then, this equilibrium is denoted by $(s^0(t^*), \dots, s^0(t^*), t^*)$, and so
\[\arg\max_{s_i\in S}u_i(t_i(s_i,t_n,s^0(t^*), \dots, s^0(t^*)),t_n,s^0(t^*), \dots, s^0(t^*))=s^0(t^*),\mathrm{for}\ i\neq n,\]
\[\arg\max_{t_n\in T}u_n(t_i(s_i,t_n,s^0(t^*), \dots, s^0(t^*)),t_n,s^0(t^*), \dots, s^0(t^*))=t^*.\]
Since any value of $t_n$ is realized by appropriately choosing $s_n$,
\begin{align*}
&\max_{t_n\in T}u_n(t_i(s_i,t_n,s^0(t^*), \dots, s^0(t^*)),t_n,s^0(t^*), \dots, s^0(t^*))\\
&=\max_{s_n\in T}u_n(t_i(s_i,s_n,s^0(t^*), \dots, s^0(t^*)),t_n(s_i, s_n,s^0(t^*),\dots,s^0(t^*)),s^0(t^*), \dots, s^0(t^*)),
\end{align*}
and
\[\arg\max_{s_n\in T}u_n(t_i(s_i,s_n,\mathbf{t}^*_l),t_n(s_i, s_n,\mathbf{t}^*_l),\mathbf{t}^*_l)=s^0(t^*).\]
Then, $(t_i(s_i,s_n,\mathbf{t}^*_l),t_n(s_i,s_n,\mathbf{t}^*_l),\mathbf{t}^*_l)$ is a Nash equilibrium, in which all players choose $s^0(t^*)$. It is equivalent to $(t^*,\dots,t^*)$.

Summarizing the results we have shown 
\begin{theorem}
Nash equilibria in the following states are equivalent.
\begin{enumerate}
	\item All players choose $t_i,\ i\in \{1, \dots, n\}$ (as their strategic variables).
	\item Some players choose $t_i$'s and the other players choose $s_i$'s.
	\item All players  choose $s_i,\ i\in \{1, \dots, n\}$.
\end{enumerate}
\end{theorem}

\section{Example of an asymmetric multi-players zero-sum game}

Consider a relative profit maximization game in an oligopoly with three firms producing differentiated goods\footnote{About relative profit maximization under imperfect competition please see \cite{mm}, \cite{ebl2}, \cite{eb2}, \cite{st}, \cite{eb1}, \cite{ebl1} and \cite{redondo}}. It is an example of multi-players zero-sum game with two strategic variables. The firms are A, B and C.  The strategic variables are the outputs and the prices of the goods of the firms.

We consider the following four cases.
\begin{enumerate}
	\item Case 1: All firms determine their outputs.

The inverse demand functions are
\[p_A=a-x_A-bx_B-bx_C,\]
\[p_B=a-x_B-bx_A-bx_C,\]
and
\[p_C=a-x_C-bx_A-bx_B,\]
where $0<b<1$. $p_A$, $p_B$ and $p_C$ are the prices of the goods of Firm A, B and C, and $x_A$, $x_B$ and $x_C$ are the outputs of them.

	\item Case 2: Firms A and B determine their outputs, and Firm C determines the price of its good.

From the inverse demand functions,
\[p_A=(1-b)a+b^2x_B-bx_B+b^2x_A-x_A+bp_C,\]
\[p_B=(1-b)a+b^2x_B-x_B+b^2x_A-bx_A+bp_C,\]
and
\[x_C=a-bx_B-bx_A-p_C\]
are derived.

	\item Case 3: Firms B and C determine the prices of their goods, and Firm A determines its output.

Also, from the above inverse demand functions, we obtain
\[p_A=\frac{(1-b)a+2b^2x_A-bx_A-x_A+bp_C+bp_B}{1+b},\]
\[x_B=\frac{(1-b)a+b^2x_A-bx_A+bp_C-p_B}{(1-b)(1+b)},\]
and
\[x_C=\frac{(1-b)a+b^2x_A-bx_A-p_C+bp_B}{(1-b)(1+b}.\]

	\item Case 4: All firms determine the prices of their goods.

From the inverse demand functions the direct demand functions are derived as follows;
\[x_A=\frac{(1-b)a-(1+b)p_A+b(p_A+p_C)}{(1-b)(1+2b)},\]
\[x_B=\frac{(1-b)a-(1+b)p_B+b(p_B+p_C)}{(1-b)(1+2b)},\]
and
\[x_C=\frac{(1-b)a-(1+b)p_C+b(p_A+p_B)}{(1-b)(1+2b)}.\]

\end{enumerate}

The (absolute) profits of the firms are 
\[\pi_A=p_Ax_A-c_Ax_A,\]
\[\pi_B=p_Bx_B-c_Bx_B,\]
and
\[\pi_C=p_Cx_C-c_Cx_C.\]
$c_A$, $c_B$ and $c_C$ are the constant marginal costs of Firm A, B and C. The relative profits of the firms are 
\[\varphi_A=\pi_A-\frac{\pi_B+\pi_C}{2},\]
\[\varphi_B=\pi_B-\frac{\pi_A+\pi_C}{2},\]
and
\[\varphi_C=\pi_C-\frac{\pi_A+\pi_B}{2}.\]
The firms determine the values of their strategic variables to maximize the relative profits. We see
\[\varphi_A+\varphi_B+\varphi_C=0,\]
so the game is zero-sum. 

We compare the equilibrium prices of the good of Firm B in four cases. Denote the value of $p_B$ in each case by $p_B^1,\ p_B^2,\ p_B^3$ and $p_B^4$. Then, we get
\[p_B^1=\frac{3bc_C-2b^2c_B+bc_B+4c_B+3bc_A+ab^2-5ab+4a}{(4-b)(b+2)},\]
\[p_B^2=\frac{A}{(4-b)(b+2)(3b+4)},\]
\[p_B^3=\frac{B}{(b+2)(b+4)(5b+4)},\]
and
\[p_B^4=\frac{3b^2c_C+3bc_C+4b^2c_B+7bc_B+4c_B+3b^2c_A+3bc_A-5ab^2+ab+4a}{(b+2)(5b+4)},\]
where
\[A=9b^2c_C+12bc_C-3b^3c_B+b^2c_B+16bc_B+16c_B-3b^3c_A+3b^2c_A+12bc_A+3ab^3-11ab^2-8ab+16a,\]
and
\begin{align*}
B=&6b^3c_C+21b^2c_C+12bc_C+b^3c_B+17b^2c_B+32bc_B+16c_B+3b^3c_A+15b^2c_A\\
&+12bc_A-5ab^3-19ab^2+8ab+16a.
\end{align*}
When $c_C=c_A$, they are
\[p_B^1=\frac{bc_B-2b^2c_B+4c_B+6bc_A+ab^2-5ab+4a}{(4-b)(b+2)},\]
\[p_B^2=\frac{b^2c_B-3b^3c_B+16bc_B+16c_B-3b^3c_A+12b^2c_A+24bc_A+3ab^3-11ab^2-8ab+16a}{(4-b)(b+2)(3b+4)},\]
\[p_B^3=\frac{b^3c_B+17b^2c_B+32bc_B+16c_B+9b^3c_A+36b^2c_A+24bc_A-5ab^3-19ab^2+8ab+16a}{(b+2)(b+4)(5b+4)},\]
and
\[p_B^4=\frac{4b^2c_B+7bc_B+4c_B+6b^2c_A+6bc_A-5ab^2+ab+4a}{(b+2)(5b+4)}.\]
Further when $c_C=c_B=c_A$, we get
\[p_B^1=p_B^2=p_B^3=p_B^4=\frac{2bc_A+c_A-ab+a}{b+2}.\]
We can show  the same result for the equilibrium prices of the goods of the other firms. Thus, in a fully symmetric game the four cases are equivalent.

It can be verified that this example with $c_A=c_B=c_C$ satisfies Assumption \ref{as1} in the sense that
\begin{quote}
the argmin (argument of the minimum) of the relative profit of Firm B with respect to the strategy of Firm A is equal to that of Firm C with the Nash equilibrium strategies of Firms B and C in Case 2 and Case 3. See (\ref{2-2}) and (\ref{2-4}).
\end{quote}

\section{Concluding Remarks}

In this paper we have shown that in a symmetric multi-players zero-sum game with two strategic variables, choice of strategic variables is irrelevant to the Nash equilibrium. In an asymmetric situation the Nash equilibrium depends on the choice of strategic variables by players other than two-players case\footnote{About two-players case please see \cite{st17}.
}. 

\section*{Acknowledgment}

This work was supported by Japan Society for the Promotion of Science KAKENHI Grant Number 15K03481  and 18K01594.

\end{document}